\documentclass[pdflatex,sn-mathphys-num]{sn-jnl}

\usepackage{graphicx}%
\usepackage{multirow}%
\usepackage{amsmath,amssymb,amsfonts}%
\usepackage{amsthm}%
\usepackage{mathrsfs}%
\usepackage[title]{appendix}%
\usepackage{xcolor}%
\usepackage{textcomp}%
\usepackage{manyfoot}%
\usepackage{booktabs}%
\usepackage{algorithm}%
\usepackage{algorithmicx}%
\usepackage{algpseudocode}%
\usepackage{listings}%
\usepackage{bm}%
\usepackage{mathtools}%

\theoremstyle{thmstyleone}%
\newtheorem{theorem}{Theorem}%
\newtheorem{proposition}[theorem]{Proposition}%
\newtheorem{lemma}[theorem]{Lemma}%
\newtheorem{corollary}[theorem]{Corollary}%

\theoremstyle{thmstyletwo}%
\newtheorem{remark}{Remark}%

\theoremstyle{thmstylethree}%
\newtheorem{definition}{Definition}%
\newtheorem{assumption}{Assumption}%

\newcommand{\M}{\mathcal M}
\newcommand{\bM}{\partial\mathcal M}

\newcommand{\B}{\mathcal B}

\newcommand{\LG}{\mathcal L_G}
\newcommand{\LP}{\mathcal L_P}
\newcommand{\Leff}{\mathcal L_\Pi}
\newcommand{\dd}{\mathrm d}

\newcommand{\Dom}{\operatorname{Dom}}
\newcommand{\spec}{\operatorname{spec}}
\newcommand{\R}{\mathcal R}

\newcommand{\Q}{\mathcal Q}

\newcommand{\HH}{\mathscr H}
\newcommand{\VV}{\mathscr V}

\newcommand{\C}{\mathcal C}
\newcommand{\J}{\mathcal J}
\newcommand{\dual}[2]{\left\langle #1,#2\right\rangle}
\newcommand{\HS}{\mathcal S}

\begin{document}

\title[Variational Openness]{Variational Openness}

\author*[1]{\fnm{Francisco} \sur{Monroy}}\email{monroy@ucm.es}

\affil*[1]{\orgdiv{Departamento de Qu\'imica F\'isica}, \orgname{Universidad Complutense de Madrid}, \orgaddress{\city{Madrid}, \country{Spain}}}

\abstract{
Variational principles in mechanics, field theory and geometric analysis are usually formulated on closed admissible classes, where boundary variations are either fixed or independently cancelled through natural boundary conditions. Variational openness is formulated here as a conservative extension of this setting. Its central premise is that stationarity requires cancellation of the total first variation, not necessarily separate cancellation of bulk and boundary contributions. Separate Euler--Lagrange and boundary equations arise only when admissible variations are independently localizable. Two regimes are distinguished. In separable open systems, bulk and boundary variations remain independently testable, and stationarity yields the usual interior equation together with an open boundary balance. In regulated open systems, admissible variations form a graph subspace in which bulk and boundary displacements are linked by a compatibility operator. Stationarity then becomes a projected balance on the admissible exchange space, allowing nontrivial bulk--boundary action exchange before total cancellation occurs. At second order, the open action defines a closed quadratic form on the admissible graph space. For pressure-like boundary couplings, the open Hessian is obtained by subtracting from the stabilizing geometric form a boundary-pressure form pulled back through the compatibility operator. A Rayleigh--Ritz criterion then yields a critical threshold at which positivity and coercivity are lost. A minimal spherical example illustrates the corresponding regulated spectral shift. The framework contains fixed-boundary, natural-boundary and classical free-boundary problems as limiting cases, while extending stationarity to regulated bulk--boundary exchange classes.
}

\keywords{calculus of variations, free boundary problems, variational principles, natural boundary conditions, Hamilton--Jacobi theory, quadratic forms, spectral instability, mathematical physics}

\pacs[MSC Classification]{49K20, 49Q10, 58E30, 35P15, 35J20, 35B35}

\maketitle

\section{Introduction}\label{sec:introduction}

Hamilton's variational principle and its field-theoretic extensions provide a standard route from stationarity of an action to equations of motion \cite{LandauLifshitz1976}. In the elementary mechanical formulation, one considers variations with fixed endpoints and requires the first variation of the action to vanish. More generally, after integration by parts, the first variation separates into an interior Euler--Lagrange term and a boundary variational flux \cite{Gelfand1963,Elsgoltz1969,Giaquinta1996,Dacorogna2008}. Standard formulations suppress the boundary contribution either by fixing the trace of the admissible variations or by imposing a natural boundary condition. In this sense, the admissible variational class is closed at the boundary \cite{Lanczos1970,Goldstein2002,Kibble2004,Arnold1989}.

The separate cancellation of interior and boundary terms is not, however, a primitive variational principle. It follows from a structural assumption on the admissible variations: the possibility of localizing bulk and boundary variations independently. Once this independence is available, the fundamental lemma of the calculus of variations separates the total first variation into an interior equation and a boundary equation. Without independent localization of admissible variations, this separation is no longer enforced by the variational principle itself.

Free-endpoint and transversality problems already show that boundary data may enter nontrivially into stationarity \cite{Gelfand1963}. Geometric and free-boundary variational problems further demonstrate that admissible boundary deformations can participate directly in variational balance laws \cite{Courant1953,Giaquinta1996}. The present work isolates the underlying structural principle. We ask what follows if the boundary sector is retained as an active variational component rather than eliminated by admissibility constraints.

We call this structure \emph{variational openness}. The term refers exclusively to the admissible variational class and does not imply stochastic forcing, dissipation or microscopic environmental coupling. An open variational problem is one for which stationarity is imposed on the total first variation,
\begin{equation}
\delta S_{\mathrm{bulk}}
+
\delta S_{\partial}
=0,
\label{eq:intro_open_stationarity}
\end{equation}
with nontrivial admissible boundary variations.

Two regimes must then be distinguished. In a \emph{separable open system}, bulk and boundary variations remain independently testable. The open principle therefore yields a bulk Euler--Lagrange equation together with a generalized natural boundary condition. In a \emph{regulated open system}, admissible variations form a proper graph subspace in which bulk and boundary displacements are linked by a compatibility operator. Stationarity is then a projected balance on the admissible graph space rather than a separate annihilation of bulk and boundary first variations. In this regime, nonzero bulk and boundary contributions may cancel only after restriction to the admissible exchange class.

The first objective of the paper is to formulate this regulated variational structure precisely. The second is to determine its second-variation and spectral consequences. Around a stationary configuration, the open action defines a quadratic form on the admissible graph space. For pressure-like boundary couplings, the corresponding open Hessian takes the form
\begin{equation}
\Q_\Pi[u]
=
\Q_G[u]
-
\Pi
\Q_P^\partial[\mathcal C u].
\label{eq:intro_open_hessian}
\end{equation}
The relevant stability directions are therefore not arbitrary bulk and boundary perturbations, but compatible exchange modes selected by the admissibility operator.

The main result is a Rayleigh--Ritz criterion for loss of positivity of the open Hessian. Under standard assumptions on the associated quadratic forms, there exists a critical threshold
\begin{equation}
\Pi_c
=
\inf_{u\neq0}
\frac{\Q_G[u]}
{\Q_P^\partial[\mathcal C u]},
\label{eq:intro_rayleigh_threshold}
\end{equation}
defined on the pressure-active admissible subspace. For \(\Pi<\Pi_c\), the open Hessian remains positive on the admissible graph space; for \(\Pi>\Pi_c\), it develops a negative direction, and, under compactness assumptions, a negative eigenvalue.

The paper is organized as follows. Section~\ref{sec:closed-open} defines closed, separable open and regulated open variational systems. Section~\ref{sec:action_exchange} formulates projected stationarity and variational action exchange. The term “exchange” refers only to cancellation structure on the admissible graph space. Section~\ref{sec:open_EL} derives the corresponding Euler--Lagrange balance laws. Section~\ref{sec:open_HJ} records the Hamilton--Jacobi representation. Section~\ref{sec:boundary_geometry} gives the geometric boundary formulation. Section~\ref{sec:open_hessian} introduces the open Hessian as a closed quadratic form on the admissible graph space. Section~\ref{sec:spectral_theorem} proves the Rayleigh--Ritz instability criterion. Section~\ref{sec:spherical_example} gives a minimal regulated spherical example. Section~\ref{sec:discussion} discusses the relation with classical variational theory and free-boundary problems.

\section{Closed, separable open and regulated open variational systems}\label{sec:closed-open}

Let \(\M\) be a smooth compact oriented \(d\)-dimensional manifold with smooth boundary \(\bM\). Let \(E\to\M\) be a smooth vector bundle and let \(\phi\) be a section of \(E\). For definiteness, we consider first-order Lagrangians,
\begin{equation}
S_{\mathrm{bulk}}[\phi]
=
\int_{\M}L(j^1\phi)\,\dd V,
\label{eq:bulk_action}
\end{equation}
where \(j^1\phi\) is the first jet of \(\phi\).

Let \(\gamma\phi\) denote the trace of \(\phi\) on \(\bM\). In the Sobolev setting,
\begin{equation}
\gamma:H^1(\M;E)\longrightarrow H^{1/2}(\bM;E|_{\bM}),
\label{eq:trace_map}
\end{equation}
and we write
\begin{equation}
\varphi:=\gamma\phi
\end{equation}
for the induced boundary field.

A closed variational problem restricts admissible variations by
\begin{equation}
\gamma v=0,
\label{eq:closed_variations}
\end{equation}
or, for free boundary values, requires the boundary variational flux to vanish through a natural boundary condition. In both cases, the boundary does not define an independent exchange channel for stationarity.

To unfold the bulk--boundary structure, let
\begin{equation}
\VV_{\mathrm{bulk}}\subset H^1(\M;E),
\qquad
\VV_{\partial}\subset H^{1/2}(\bM;E|_{\bM})
\label{eq:bulk_boundary_test_spaces}
\end{equation}
be the bulk and boundary test spaces, and define
\begin{equation}
\VV_{\mathrm{prod}}
:=
\VV_{\mathrm{bulk}}\oplus\VV_{\partial}.
\label{eq:product_variation_space}
\end{equation}
Elements of \(\VV_{\mathrm{prod}}\) are denoted \((v,w)\), where \(v\) is a bulk variation and \(w\) is a boundary variation. This product representation is an unfolded form of the variational derivative; admissibility is imposed by selecting either independent tests or a constrained graph.

\begin{definition}[Separable open variational system]\label{def:separable_open_system}
A separable open variational system has action
\begin{equation}
S[\phi]
=
S_{\mathrm{bulk}}[\phi]
+
S_{\partial}[\varphi],
\qquad
\varphi=\gamma\phi,
\label{eq:open_action}
\end{equation}
with
\begin{equation}
S_{\partial}[\varphi]
=
\int_{\bM}\B(j^1_{\partial}\varphi,h)\,\dd A_h.
\label{eq:boundary_action}
\end{equation}
Stationarity is imposed for admissible variations whose bulk and boundary components are independently testable. No fixed-trace condition is imposed a priori.
\end{definition}

\begin{definition}[Regulated open variational system]\label{def:regulated_open_system}
A regulated open variational system is an open variational problem in which admissible bulk and boundary variations are linked by a bounded linear compatibility operator
\begin{equation}
\C:\VV_{\mathrm{bulk}}\longrightarrow \VV_{\partial}.
\label{eq:compatibility_operator}
\end{equation}
The admissible variation space is the graph
\begin{equation}
\VV_{\mathrm{op}}
:=
\operatorname{Graph}(\C)
=
\{(v,\C v):v\in\VV_{\mathrm{bulk}}\}
\subset
\VV_{\mathrm{prod}}.
\label{eq:open_graph_space}
\end{equation}
It is equipped with the graph norm
\begin{equation}
\|(v,\C v)\|_{\VV_{\mathrm{op}}}^{2}
=
\|v\|_{\VV_{\mathrm{bulk}}}^{2}
+
\|\C v\|_{\VV_{\partial}}^{2}.
\label{eq:graph_norm}
\end{equation}
Stationarity is imposed only after restriction to this graph:
\begin{equation}
\delta S_{\mathrm{bulk}}[\phi](v)
+
\delta S_{\partial}[\varphi](\C v)
=
0,
\qquad
\varphi=\gamma\phi,
\quad
\forall (v,\C v)\in\VV_{\mathrm{op}}.
\label{eq:regulated_open_stationarity}
\end{equation}
\end{definition}

The boundedness of \(\C\) makes \(\VV_{\mathrm{op}}\) a closed admissible subspace of \(\VV_{\mathrm{prod}}\) in the standard graph-norm sense \cite{Conway1990,Brezis2011}. This is the regulation condition: the boundary sector cannot cancel arbitrary bulk variations, but only those transmitted through the specified compatibility channel.

Closed, separable and regulated problems correspond to different admissible classes. If \(\C=0\), the unfolded boundary component is suppressed; the usual fixed-boundary problem is recovered when \(\VV_{\mathrm{bulk}}\) is additionally chosen with fixed trace, \(\gamma v=0\). If boundary trace directions are independently testable, one obtains the separable natural-boundary case. If \(\C\) defines a nontrivial graph, stationarity is a projected condition on \(\operatorname{Graph}(\C)\), not separate annihilation on \(\VV_{\mathrm{bulk}}\) and \(\VV_{\partial}\).

\section{Action exchange and projected stationarity}\label{sec:action_exchange}

The unfolded first variation is a functional on
\(\VV_{\mathrm{prod}}=\VV_{\mathrm{bulk}}\oplus\VV_{\partial}\). For \((v,w)\in\VV_{\mathrm{prod}}\), write
\begin{equation}
\delta S[\phi](v,w)
=
\J_{\mathrm{bulk}}[\phi](v)
+
\J_{\partial}[\phi](w),
\label{eq:first_variation_product_split}
\end{equation}
with
\begin{equation}
\J_{\mathrm{bulk}}[\phi](v)
:=
\int_{\M}\dual{E_L(\phi)}{v}\,\dd V,
\label{eq:Jbulk_def}
\end{equation}
and
\begin{equation}
\J_{\partial}[\phi](w)
:=
\int_{\bM}
\dual{
\Pi_L(\phi)+E_{\partial\B}(\gamma\phi)
}{w}
\,\dd A_h,
\label{eq:Jboundary_def}
\end{equation}
up to corner terms. In a regulated open system the admissible pairs are restricted to
\[
(v,w)=(v,\C v).
\]
Thus the first variation pulled back to the graph is
\begin{equation}
\mathfrak F_{\phi}(v)
:=
\delta S[\phi](v,\C v)
=
\J_{\mathrm{bulk}}[\phi](v)
+
\J_{\partial}[\phi](\C v).
\label{eq:projected_variational_functional}
\end{equation}
Regulated stationarity is
\begin{equation}
\mathfrak F_{\phi}=0
\quad
\text{in }\VV_{\mathrm{bulk}}'.
\label{eq:projected_stationarity}
\end{equation}

Since \(\C:\VV_{\mathrm{bulk}}\to\VV_{\partial}\) is bounded, its adjoint
\begin{equation}
\C^*:\VV_{\partial}'\longrightarrow \VV_{\mathrm{bulk}}'
\label{eq:compatibility_adjoint}
\end{equation}
is well defined. Let
\begin{equation}
B_{\partial}(\phi)
:=
\Pi_L(\phi)+E_{\partial\B}(\gamma\phi)
\in \VV_{\partial}'
\label{eq:boundary_force_functional}
\end{equation}
be the total boundary variational flux. Then
\begin{equation}
\J_{\partial}[\phi](\C v)
=
\dual{B_{\partial}(\phi)}{\C v}_{\partial}
=
\dual{\C^*B_{\partial}(\phi)}{v}_{\mathrm{bulk}}.
\label{eq:adjoint_pullback_boundary}
\end{equation}
Therefore regulated stationarity is equivalently
\begin{equation}
E_L(\phi)
+
\C^*
\left[
\Pi_L(\phi)+E_{\partial\B}(\gamma\phi)
\right]
=
0
\quad
\text{in }\VV_{\mathrm{bulk}}'.
\label{eq:projected_EL_adjoint}
\end{equation}
This is the projected Euler--Lagrange balance. The boundary flux contributes only after pullback by the admissibility operator.

\begin{lemma}[Failure of separate cancellation]\label{lem:failure_separate_cancellation}
Let
\[
\VV_{\mathrm{op}}
=
\operatorname{Graph}(\C)
\subset
\VV_{\mathrm{bulk}}\oplus\VV_{\partial}
\]
be a proper graph subspace. If \(\VV_{\mathrm{op}}\) does not contain arbitrary independent directions \((v,0)\) and \((0,w)\), then
\begin{equation}
\J_{\mathrm{bulk}}(v)+\J_{\partial}(\C v)=0
\quad
\forall v\in\VV_{\mathrm{bulk}}
\label{eq:graph_stationarity_condition}
\end{equation}
does not imply, in general,
\begin{equation}
\J_{\mathrm{bulk}}=0,
\qquad
\J_{\partial}=0
\label{eq:separate_cancellation_false}
\end{equation}
on \(\VV_{\mathrm{bulk}}\) and \(\VV_{\partial}\), respectively.
\end{lemma}

\begin{proof}
Equation~\eqref{eq:graph_stationarity_condition} states that the product functional
\[
(v,w)\mapsto \J_{\mathrm{bulk}}(v)+\J_{\partial}(w)
\]
vanishes only after restriction to \(w=\C v\). Since \((v,0)\) and \((0,w)\) are not arbitrary admissible tests on the graph, the two components cannot be separated by the fundamental lemma. Equivalently, a nonzero pair satisfying
\[
\J_{\mathrm{bulk}}=-\C^*\J_{\partial}
\quad\text{in }\VV_{\mathrm{bulk}}'
\]
has zero total variation on \(\operatorname{Graph}(\C)\) without requiring either component to vanish on the unconstrained product space.
\end{proof}

\begin{corollary}[Nontrivial exchange]\label{cor:nontrivial_exchange}
If \(B_{\partial}\neq0\) and
\begin{equation}
E_L=-\C^*B_{\partial}
\quad
\text{in }\VV_{\mathrm{bulk}}',
\label{eq:nontrivial_exchange_condition}
\end{equation}
then the total first variation vanishes on \(\operatorname{Graph}(\C)\), although the bulk and boundary variational contributions need not vanish separately.
\end{corollary}

\begin{definition}[Variational action exchange]\label{def:action_exchange}
A stationary regulated open configuration exhibits variational action exchange if, for some admissible graph variation \((v,\C v)\),
\begin{equation}
\J_{\mathrm{bulk}}[\phi](v)\neq0,
\qquad
\J_{\partial}[\phi](\C v)\neq0,
\qquad
\J_{\mathrm{bulk}}[\phi](v)
=
-
\J_{\partial}[\phi](\C v).
\label{eq:mutual_cancellation}
\end{equation}
\end{definition}

Thus regulated action exchange is the restricted stationarity condition
\begin{equation}
(S_{\mathrm{bulk}}+S_{\partial})'|_{\operatorname{Graph}(\C)}=0,
\label{eq:restricted_total_variation}
\end{equation}
rather than separate stationarity of the two terms.

\section{Open Euler--Lagrange balance}\label{sec:open_EL}

We now record the Euler--Lagrange consequences of separable and regulated admissibility. The bulk first variation has the standard decomposition
\begin{equation}
\delta S_{\mathrm{bulk}}[\phi](v)
=
\int_{\M}\dual{E_L(\phi)}{v}\,\dd V
+
\int_{\bM}\dual{\Pi_L(\phi)}{\gamma v}\,\dd A_h,
\label{eq:bulk_variation}
\end{equation}
where \(E_L(\phi)\) is the Euler--Lagrange expression and \(\Pi_L(\phi)\) is the boundary variational flux. The boundary action satisfies
\begin{equation}
\delta S_{\partial}[\varphi](w)
=
\int_{\bM}\dual{E_{\partial\B}(\varphi)}{w}\,\dd A_h
+
\int_{\partial(\bM)}\Theta_{\partial\B}(\varphi,w),
\label{eq:boundary_variation}
\end{equation}
with \(w\in\VV_{\partial}\). We assume that the corner term either vanishes or is cancelled by admissibility conditions or by an additional corner functional.

For a general unfolded pair \((v,w)\in\VV_{\mathrm{bulk}}\oplus\VV_{\partial}\),
\begin{equation}
\delta S[\phi](v,w)
=
\int_{\M}\dual{E_L(\phi)}{v}\,\dd V
+
\int_{\bM}
\dual{
\Pi_L(\phi)+E_{\partial\B}(\gamma\phi)
}{w}
\,\dd A_h.
\label{eq:total_first_variation_product}
\end{equation}
The equations obtained from \eqref{eq:total_first_variation_product} depend on the admissible test class.

\begin{theorem}[Separable open Euler--Lagrange equations]\label{thm:separable_open_EL}
Let \(\phi\) be a stationary point of the separable open action \eqref{eq:open_action}. If bulk and boundary variations are independently localizable, then
\begin{align}
E_L(\phi)&=0 &&\text{in }\M,
\label{eq:bulk_EL}\\
\Pi_L(\phi)+E_{\partial\B}(\gamma\phi)&=0 &&\text{on }\bM.
\label{eq:open_boundary_condition}
\end{align}
Conversely, if \eqref{eq:bulk_EL}--\eqref{eq:open_boundary_condition} hold and the corner contribution in \eqref{eq:boundary_variation} vanishes, then \(\delta S[\phi](v,w)=0\) for all separable admissible variations.
\end{theorem}

\begin{proof}
Independent tests of the form \((v,0)\), with \(v\) compactly supported in \(\M\), give \eqref{eq:bulk_EL}. Independent boundary tests \((0,w)\) give \eqref{eq:open_boundary_condition}. The converse follows by substitution in \eqref{eq:total_first_variation_product}.
\end{proof}

\begin{theorem}[Regulated open Euler--Lagrange balance]\label{thm:regulated_open_EL}
Let \(\phi\) be stationary on the regulated graph
\[
\VV_{\mathrm{op}}
=
\operatorname{Graph}(\C).
\]
Then
\begin{equation}
\int_{\M}\dual{E_L(\phi)}{v}\,\dd V
+
\int_{\bM}
\dual{
\Pi_L(\phi)+E_{\partial\B}(\gamma\phi)
}{\C v}
\,\dd A_h
=
0
\quad
\forall v\in\VV_{\mathrm{bulk}}.
\label{eq:regulated_EL_balance}
\end{equation}
Equivalently,
\begin{equation}
E_L(\phi)
+
\C^*
\left[
\Pi_L(\phi)+E_{\partial\B}(\gamma\phi)
\right]
=
0
\quad
\text{in }\VV_{\mathrm{bulk}}'.
\label{eq:regulated_EL_adjoint_balance}
\end{equation}
\end{theorem}

\begin{proof}
On the graph, \(w=\C v\). Substitution into \eqref{eq:total_first_variation_product} gives \eqref{eq:regulated_EL_balance}. Since \(\C\) is bounded, \(\C^*:\VV_{\partial}'\to\VV_{\mathrm{bulk}}'\) is well defined, and the boundary term is the dual pairing of \(\C^*[\Pi_L+E_{\partial\B}]\) with \(v\). This gives \eqref{eq:regulated_EL_adjoint_balance}.
\end{proof}

\begin{corollary}[Exchange form]\label{cor:exchange_EL}
In a regulated open system,
\begin{equation}
E_L(\phi)
=
-
\C^*
\left[
\Pi_L(\phi)+E_{\partial\B}(\gamma\phi)
\right]
\quad
\text{in }\VV_{\mathrm{bulk}}'.
\label{eq:EL_exchange_form}
\end{equation}
Thus the Euler--Lagrange residual is not required to vanish separately; it is balanced by the boundary variational flux pulled back through the admissibility operator.
\end{corollary}

\begin{remark}[Limits]\label{rem:closed_natural_exchange_limits}
If \(\C=0\), the fixed-boundary problem is recovered. If boundary directions are independently testable, the regulated balance separates into \eqref{eq:bulk_EL}--\eqref{eq:open_boundary_condition}. For a nontrivial graph, stationarity remains the projected exchange condition \eqref{eq:EL_exchange_form}.
\end{remark}

\section{Open Hamilton--Jacobi representation}\label{sec:open_HJ}

The open variational principle also admits a Hamilton--Jacobi representation. No additional dynamical assumption is introduced; this section rewrites the first-variation balance of Section~\ref{sec:open_EL} at the level of Hamilton's principal function.

Consider the finite-dimensional action
\begin{equation}
S[q]
=
\int_{t_0}^{t_1}L(q,\dot q,t)\,\dd t
+
B(q(t_1),t_1),
\label{eq:HJ_open_action}
\end{equation}
with fixed initial endpoint. Let \(\HS(q,t)\) be Hamilton's principal function. In the interior it satisfies
\begin{equation}
\partial_t\HS(q,t)
+
H(q,\nabla_q\HS,t)
=
0.
\label{eq:HJ_bulk}
\end{equation}
For freely testable terminal variations, the natural endpoint condition is
\begin{equation}
D_q\HS(q,t_1)
+
D_q B(q,t_1)
=
0.
\label{eq:HJ_open_boundary}
\end{equation}
This is the Hamilton--Jacobi counterpart of the separable boundary condition
\eqref{eq:open_boundary_condition}.

For a regulated endpoint class, terminal variations are restricted by
\begin{equation}
\delta q(t_1)=\C v.
\label{eq:HJ_endpoint_compatibility}
\end{equation}
The endpoint stationarity condition is
\begin{equation}
\dual{D_q\HS+D_qB}{\C v}
=
0
\quad
\text{for all admissible }v,
\label{eq:HJ_regulated_balance}
\end{equation}
or, equivalently,
\begin{equation}
\C^*
\left(
D_q\HS+D_qB
\right)
=
0
\quad
\text{on the admissible endpoint space}.
\label{eq:HJ_regulated_adjoint_balance}
\end{equation}
Thus the regulated Hamilton--Jacobi condition is the pullback analogue of
\eqref{eq:regulated_EL_adjoint_balance}.

In field form, let \(\HS[\Sigma,\varphi]\) be the on-shell action evaluated on a boundary hypersurface \(\Sigma\) with boundary field \(\varphi=\gamma\phi\). Its first variation with respect to the boundary datum is denoted
\[
D_{\varphi}\HS[\Sigma,\varphi]\in\VV_{\partial}'.
\]
The separable open Hamilton--Jacobi boundary condition is
\begin{equation}
D_{\varphi}\HS[\Sigma,\varphi]
+
D_{\varphi}S_{\partial}[\varphi]
=
0
\quad\text{in }\VV_{\partial}'.
\label{eq:functional_HJ_open_boundary}
\end{equation}
For a regulated admissible graph \(w=\C v\), this becomes
\begin{equation}
\dual{
D_{\varphi}\HS[\Sigma,\varphi]
+
D_{\varphi}S_{\partial}[\varphi]
}{\C v}
=
0
\quad
\forall v\in\VV_{\mathrm{bulk}},
\label{eq:functional_HJ_regulated_boundary}
\end{equation}
or
\begin{equation}
\C^*
\left(
D_{\varphi}\HS[\Sigma,\varphi]
+
D_{\varphi}S_{\partial}[\varphi]
\right)
=
0
\quad
\text{in }\VV_{\mathrm{bulk}}'.
\label{eq:functional_HJ_adjoint_boundary}
\end{equation}
Equation~\eqref{eq:functional_HJ_adjoint_boundary} is the Hamilton--Jacobi form of regulated action exchange.

\section{Open variation of the boundary geometry}\label{sec:boundary_geometry}

For geometric variational problems, openness may also act through the induced boundary metric \(h_{ab}\). The boundary part of the first variation contains
\begin{equation}
\delta S
=
\cdots+
\frac12
\int_{\bM}
\sqrt{|h|}\,
\mathcal T^{ab}\delta h_{ab}
\,\dd^{d-1}x,
\label{eq:metric_variation}
\end{equation}
where
\begin{equation}
\mathcal T^{ab}
:=
\frac{2}{\sqrt{|h|}}
\frac{\delta S}{\delta h_{ab}}
\label{eq:boundary_stress_def}
\end{equation}
is the boundary variational stress. If \(S=S_G+S_{\partial}\), then
\begin{equation}
\mathcal T^{ab}
=
\mathcal T_G^{ab}
+
\mathcal T_{\partial}^{ab}.
\label{eq:stress_split}
\end{equation}

\begin{proposition}[Projected geometric boundary balance]\label{prop:projected_geometric_balance}
If boundary metric variations are independently testable, stationarity gives the pointwise balance
\begin{equation}
\mathcal T_G^{ab}
+
\mathcal T_{\partial}^{ab}
=
0
\quad
\text{on }\bM.
\label{eq:stress_balance_separable}
\end{equation}
If, instead, admissible boundary metric variations are induced by bulk variations through a bounded operator
\begin{equation}
\C_h:\VV_{\mathrm{bulk}}\to\mathscr S_{\partial},
\qquad
\delta h_{ab}=(\C_h v)_{ab},
\label{eq:metric_compatibility_operator}
\end{equation}
then stationarity is the projected balance
\begin{equation}
\J_{\mathrm{bulk}}^{\mathrm{geom}}[v]
+
\frac12
\int_{\bM}
\sqrt{|h|}
\left(
\mathcal T_G^{ab}
+
\mathcal T_{\partial}^{ab}
\right)
(\C_h v)_{ab}
\,\dd^{d-1}x
=
0
\quad
\forall v\in\VV_{\mathrm{bulk}}.
\label{eq:regulated_stress_balance}
\end{equation}
Equivalently,
\begin{equation}
\J_{\mathrm{bulk}}^{\mathrm{geom}}
+
\frac12
\C_h^*
\left[
\sqrt{|h|}
\left(
\mathcal T_G+\mathcal T_{\partial}
\right)
\right]
=
0
\quad
\text{in }\VV_{\mathrm{bulk}}'.
\label{eq:regulated_stress_adjoint_balance}
\end{equation}
\end{proposition}

\begin{proof}
The separable case follows by testing \eqref{eq:metric_variation} with arbitrary \(\delta h_{ab}\). In the regulated case, substitute \(\delta h_{ab}=(\C_h v)_{ab}\). Boundedness of \(\C_h\) gives the adjoint map \(\C_h^*:\mathscr S_{\partial}'\to\VV_{\mathrm{bulk}}'\), yielding \eqref{eq:regulated_stress_adjoint_balance}.
\end{proof}

Thus boundary geometry obeys the same structure as the field variation: pointwise balance is recovered only when boundary metric variations are independently testable; otherwise the stress balance is pulled back to the admissible bulk deformation space.

\section{The open Hessian as a quadratic form}\label{sec:open_hessian}

Let \(\phi_0\) be a stationary configuration of the open action, and let
\[
\VV_{\mathrm{op}}
=
\operatorname{Graph}(\C)
\]
be the regulated admissible variation space, after imposing the bulk--boundary compatibility relation and quotienting null directions. We regard \(\VV_{\mathrm{op}}\) as a Hilbert graph space, with the functional-analytic structure inherited from \(\VV_{\mathrm{bulk}}\oplus\VV_{\partial}\) \cite{Conway1990,Brezis2011}. The second variation restricted to this graph defines
\begin{equation}
\delta^2S[\phi_0]\bigl((u,\C u),(u,\C u)\bigr)
=
\frac12\Q_{\mathrm{op}}[u],
\qquad
(u,\C u)\in\VV_{\mathrm{op}}.
\label{eq:second_variation_graph_form}
\end{equation}
Thus \(\Q_{\mathrm{op}}\) is the pullback of the full Hessian to the admissible exchange space.

In general,
\begin{equation}
\Q_{\mathrm{op}}[u]
=
\Q_{\mathrm{bulk}}[u]
+
\Q_{\partial}[\C u]
+
2\Q_{\mathrm{mix}}[u,\C u],
\label{eq:general_open_quadratic_form}
\end{equation}
where \(\Q_{\mathrm{mix}}\) collects mixed second-variation terms. For an additively split action with fixed compatibility operator, this reduces to
\begin{equation}
\Q_{\mathrm{op}}[u]
=
\Q_{\mathrm{bulk}}[u]
+
\Q_{\partial}[\C u].
\label{eq:additive_open_quadratic_form}
\end{equation}

The pressure-like case considered below has the form
\begin{equation}
\Q_{\Pi}[u]
=
\Q_G[u]
-
\Pi\Q_P^{\partial}[\C u],
\label{eq:pressure_open_quadratic_form}
\end{equation}
where \(\Q_G\) is the stabilizing bulk or geometric form and \(\Q_P^{\partial}\) is a non-negative boundary-pressure form. For compactness we write
\begin{equation}
\Q_P[u]
:=
\Q_P^{\partial}[\C u],
\qquad
\Q_\Pi[u]
=
\Q_G[u]-\Pi\Q_P[u],
\label{eq:open_quadratic_form}
\end{equation}
always understood on \(\VV_{\mathrm{op}}\).

\begin{assumption}[Quadratic-form setting]\label{ass:quadratic_forms}
The following hypotheses are assumed in Sections~\ref{sec:open_hessian}--\ref{sec:spectral_theorem}.
\begin{enumerate}
\item \(\VV_{\mathrm{op}}\) is a dense subspace of a Hilbert space \(\HH\), or a closed admissible graph subspace after quotienting null modes.
\item \(\Q_G:\VV_{\mathrm{op}}\to\mathbb R\) is symmetric, closed and bounded from below.
\item \(\Q_G\) is coercive on \(\VV_{\mathrm{op}}\): there exists \(c_G>0\) such that
\begin{equation}
\Q_G[u]\geq c_G\|u\|_{\VV_{\mathrm{op}}}^{2}.
\label{eq:QG_coercive}
\end{equation}
\item \(\Q_P:\VV_{\mathrm{op}}\to\mathbb R\) is symmetric, non-negative, and either \(\Q_G\)-bounded with relative bound strictly smaller than one or compact with respect to the \(\Q_G\)-form norm.
\item There exists \(u\in\VV_{\mathrm{op}}\) such that \(\Q_P[u]>0\).
\end{enumerate}
\end{assumption}

Under Assumption~\ref{ass:quadratic_forms}, \(\Q_\Pi\) is a closed semibounded form in the range of \(\Pi\) considered. By the standard representation theorem for closed semibounded quadratic forms, it defines a unique self-adjoint operator \(\Leff\) \cite{Kato1995,ReedSimon1978}:
\begin{equation}
\Q_\Pi[u,v]
=
\dual{u}{\Leff v}_{\HH},
\qquad
v\in\Dom(\Leff)\subset\VV_{\mathrm{op}},
\quad
u\in\VV_{\mathrm{op}}.
\label{eq:hessian_representation}
\end{equation}
When the forms are represented by operators, this is written formally as
\begin{equation}
\Leff
=
\LG-\Pi\LP.
\label{eq:open_hessian}
\end{equation}
This expression is an operator representation of the pulled-back graph form \eqref{eq:open_quadratic_form}, not an unconstrained subtraction of independent bulk and boundary Hessians.

\begin{definition}[Open Hessian]\label{def:open_hessian}
The open Hessian at a stationary configuration is the self-adjoint operator associated with the closed quadratic form \(\Q_\Pi\) obtained by restricting the second variation of the total action to \(\VV_{\mathrm{op}}\). Positivity and coercivity are defined on this admissible exchange space.
\end{definition}

The spectrum of \(\Leff\) is therefore the spectrum of the Hessian projected onto admissible exchange directions. Modes outside the graph of \(\C\) are not tested by the regulated variational problem.

\section{Spectral instability theorem}\label{sec:spectral_theorem}

We now state the instability criterion on the regulated admissible graph. Throughout this section,
\[
\Q_P[u]\equiv \Q_P^{\partial}[\C u],
\qquad
u\in\VV_{\mathrm{op}}.
\]
Thus all Rayleigh quotients are computed only over admissible exchange directions.

Define the pressure-active cone
\begin{equation}
\VV_P
:=
\{u\in\VV_{\mathrm{op}}:\Q_P[u]>0\},
\label{eq:pressure_active_cone}
\end{equation}
and, for \(u\in\VV_P\),
\begin{equation}
\R[u]
=
\frac{\Q_G[u]}{\Q_P[u]}.
\label{eq:rayleigh_quotient}
\end{equation}
The critical pressure is
\begin{equation}
\Pi_c
:=
\inf_{u\in\VV_P}\R[u]
=
\inf_{u:\,\Q_P^{\partial}[\C u]>0}
\frac{\Q_G[u]}{\Q_P^{\partial}[\C u]}.
\label{eq:critical_pressure}
\end{equation}

\begin{theorem}[Regulated boundary-pressure loss of positivity]
\label{thm:boundary_pressure_instability}

Assume Assumption~\ref{ass:quadratic_forms}. Let
\[
\Q_\Pi=\Q_G-\Pi\Q_P
\]
be defined on \(\VV_{\mathrm{op}}\), and let \(\Leff\) be the self-adjoint operator associated with \(\Q_\Pi\). Then:
\begin{align}
\Pi<\Pi_c
&\Rightarrow
\Q_\Pi[u]>0
\quad
\forall u\neq0,
\label{eq:positive_regime}\\
\Pi=\Pi_c
&\Rightarrow
\inf_{u\in\VV_P}
\frac{\Q_\Pi[u]}{\Q_P[u]}
=
0,
\label{eq:marginal_regime}\\
\Pi>\Pi_c
&\Rightarrow
\exists u_*\in\VV_{\mathrm{op}}
\text{ such that }
\Q_\Pi[u_*]<0.
\label{eq:negative_regime}
\end{align}
Consequently, for \(\Pi>\Pi_c\), the stationary configuration is not a strict local minimum relative to the regulated exchange space. If, in addition, the embedding of the \(\Q_G\)-form domain into \(\HH\) is compact, or more generally if
\[
\inf\spec(\Leff)
\]
is attained as an eigenvalue, then \(\Leff\) has a negative eigenvalue on \(\VV_{\mathrm{op}}\) \cite{Kato1995,ReedSimon1978,Davies1995}.
\end{theorem}

\begin{proof}
For any \(u\in\VV_{\mathrm{op}}\),
\begin{equation}
\Q_\Pi[u]
=
\Q_G[u]
-
\Pi\Q_P[u].
\label{eq:QP_identity}
\end{equation}
If \(\Q_P[u]=0\), coercivity gives
\[
\Q_\Pi[u]=\Q_G[u]>0
\qquad
(u\neq0).
\]
If \(\Q_P[u]>0\), then
\begin{equation}
\Q_\Pi[u]
=
\Q_P[u]\bigl(\R[u]-\Pi\bigr).
\label{eq:rayleigh_factorization}
\end{equation}
For \(\Pi<\Pi_c\), one has \(\R[u]\geq\Pi_c>\Pi\), hence \(\Q_\Pi[u]>0\). At \(\Pi=\Pi_c\), the normalized infimum vanishes. For \(\Pi>\Pi_c\), the definition of the infimum yields \(u_*\in\VV_P\) such that
\[
\R[u_*]<\Pi,
\]
and therefore
\[
\Q_\Pi[u_*]<0.
\]
This proves loss of positivity of the second variation on \(\VV_{\mathrm{op}}\). If \(\inf\spec(\Leff)\) is attained as an eigenvalue, the min--max principle implies that this eigenvalue is negative.
\end{proof}

\begin{corollary}[Loss of coercivity on the exchange space]
\label{cor:coercivity}

Under the hypotheses of Theorem~\ref{thm:boundary_pressure_instability}, if \(\Pi>\Pi_c\), the quadratic form \(\Q_\Pi\) is not coercive on \(\VV_{\mathrm{op}}\). If \(\Pi\nearrow\Pi_c\) from below along a compact minimizing branch, the coercivity constant tends to zero.
\end{corollary}

\begin{remark}[Regulation of the instability channel]
\label{rem:regulated_instability_channel}

The threshold \(\Pi_c\) is not computed on the unconstrained boundary space. It is computed only over boundary perturbations reachable as \(\C u\). Thus \(\C\) regulates both the first-order exchange channel and the second-order instability channel.
\end{remark}

\begin{remark}[Negative direction versus negative eigenvalue]
\label{rem:negative_direction}

Equation~\eqref{eq:negative_regime} proves loss of positivity of the second variation on the admissible graph. The stronger statement that \(\Leff\) possesses a negative eigenvalue requires additional spectral hypotheses ensuring that \(\inf\spec(\Leff)\) is attained as an eigenvalue.
\end{remark}

\section{Minimal spherical example}\label{sec:spherical_example}

We give a minimal example illustrating how the compatibility operator regulates the spectral instability channel. Let \(\bM=S_R^2\), and let
\[
u=\sum_{\ell,m}u_{\ell m}Y_{\ell m}
\]
be a normal deformation. Consider the stabilizing quadratic form
\begin{equation}
\Q_G[u]
=
\sigma_G
\int_{S_R^2}
\left(
|\nabla_\Sigma u|^2
+
\frac{2}{R^2}u^2
\right)
\dd A,
\qquad
\sigma_G>0,
\label{eq:spherical_QG}
\end{equation}
and the boundary pressure form
\begin{equation}
\Q_P^{\partial}[w]
=
\frac{2}{R}
\int_{S_R^2}w^2\,\dd A.
\label{eq:spherical_boundary_QP}
\end{equation}

Let the compatibility operator act diagonally on spherical harmonics,
\begin{equation}
\C Y_{\ell m}
=
c_\ell Y_{\ell m}.
\label{eq:spherical_compatibility_operator}
\end{equation}
Then the admissible boundary displacement is
\[
w=\C u
=
\sum_{\ell,m}c_\ell u_{\ell m}Y_{\ell m},
\]
and the pulled-back pressure form is
\begin{equation}
\Q_P[u]
=
\Q_P^{\partial}[\C u]
=
\frac{2}{R}
\sum_{\ell,m}
|c_\ell|^2|u_{\ell m}|^2.
\label{eq:spherical_pulled_pressure_form}
\end{equation}

Using
\[
-\Delta_\Sigma Y_{\ell m}
=
\frac{\ell(\ell+1)}{R^2}Y_{\ell m},
\]
the regulated open quadratic form diagonalizes as
\begin{equation}
\Q_\Pi[u]
=
\sum_{\ell,m}
\left[
\sigma_G
\frac{\ell(\ell+1)+2}{R^2}
-
\frac{2\Pi}{R}|c_\ell|^2
\right]
|u_{\ell m}|^2.
\label{eq:spherical_regulated_open_form}
\end{equation}
Thus the regulated eigenvalue shift is
\begin{equation}
\lambda_\ell^{\C}(\Pi)
=
\sigma_G
\frac{\ell(\ell+1)+2}{R^2}
-
\frac{2\Pi}{R}|c_\ell|^2.
\label{eq:spherical_regulated_eigenvalues}
\end{equation}
The mode-dependent threshold is
\begin{equation}
\Pi_{c,\ell}^{\C}
=
\frac{\sigma_G}{2R}
\frac{\ell(\ell+1)+2}{|c_\ell|^2},
\qquad c_\ell\neq0.
\label{eq:spherical_regulated_threshold_mode}
\end{equation}
Modes with \(c_\ell=0\) are not coupled to the boundary-pressure channel. Therefore the regulated critical pressure is
\begin{equation}
\Pi_c^{\C}
=
\inf_{\ell\in\mathcal I_{\mathrm{adm}},\,c_\ell\neq0}
\frac{\sigma_G}{2R}
\frac{\ell(\ell+1)+2}{|c_\ell|^2},
\label{eq:spherical_regulated_threshold}
\end{equation}
where \(\mathcal I_{\mathrm{adm}}\) denotes the harmonics retained after constraints such as volume fixing or quotienting of rigid modes.

This example shows explicitly that \(\C\) does more than restrict admissible variations: it selects which boundary modes can exchange action with the bulk and therefore which spectral channels can lose positivity.

\section{Discussion}\label{sec:discussion}

The construction developed here changes the role of the boundary in variational mechanics. In the closed formulation, the boundary is removed from the admissible variation class by fixing traces or by imposing independently testable natural conditions. In the regulated open formulation, the boundary remains part of the variational system, but only through an admissible graph
\[
\VV_{\mathrm{op}}
=
\operatorname{Graph}(\C).
\]
This is the main structural distinction: openness is neither arbitrary boundary forcing nor unconstrained nonlocal cancellation. It is stationarity after restriction to a specified compatibility channel.

The compatibility operator \(\C\) is therefore the object that regulates bulk--boundary exchange. It determines which boundary variations are coupled to a given bulk variation, and its adjoint pulls boundary variational fluxes back to the bulk test space. The projected balance
\[
E_L
+
\C^*
\left(
\Pi_L+E_{\partial\B}
\right)
=
0
\quad
\text{in }\VV_{\mathrm{bulk}}'
\]
is the weak form of this exchange. It does not modify the Euler--Lagrange operator by an external force; rather, it states that the total first variation vanishes on the admissible graph.

This viewpoint clarifies the relation with classical boundary conditions. Fixed-boundary problems correspond to a trivial boundary channel. Natural boundary conditions correspond to independently testable boundary directions. Regulated open problems lie between these limits: boundary degrees of freedom are active, but only through the operator that defines admissibility. Thus the generalized natural condition is not discarded; it is recovered as the separable limit of graph-space stationarity.

The Hamilton--Jacobi and geometric stress formulations follow the same pattern. In action space, the boundary derivative of the on-shell action is pulled back by \(\C^*\). In geometric problems, boundary stress is either balanced pointwise when metric variations are independent, or pulled back through a metric compatibility operator when they are not. These formulations are not additional assumptions; they are equivalent representations of the same regulated admissible class.

At second order, the decisive object is not the unconstrained Hessian but its restriction to \(\VV_{\mathrm{op}}\). For pressure-like boundary couplings, this produces the pulled-back form
\[
\Q_\Pi[u]
=
\Q_G[u]
-
\Pi\Q_P^{\partial}[\C u].
\]
The Rayleigh threshold is therefore computed only over exchange-compatible directions. This explains why the compatibility operator controls both first-order action exchange and second-order loss of positivity.

The spherical example makes this control explicit. When \(\C\) is diagonal in spherical harmonics, the coefficients \(c_\ell\) select which modes couple to the boundary-pressure sector. Modes outside the range of \(\C\) do not enter the instability quotient; modes with stronger coupling destabilize at lower pressure. The example is not intended as a model of a specific elastic surface, but as a transparent realization of the general graph-space criterion.

The framework remains close to classical variational analysis. It uses standard first variations, natural boundary terms, closed quadratic forms and Rayleigh--Ritz arguments \cite{Courant1953,ReedSimon1978,Evans2010}. The new ingredient is the admissible exchange space. Once that space is specified, the rest of the construction follows by restriction, pullback and spectral analysis on the resulting graph.

\section{Conclusion}\label{sec:conclusion}

Variational openness has been formulated as a conservative extension of closed variational mechanics. The extension leaves the variational calculus unchanged, but changes the admissible class on which stationarity is imposed. Its central object is the regulated graph
\[
\VV_{\mathrm{op}}
=
\operatorname{Graph}(\C),
\]
which encodes the admissible exchange channel between bulk and boundary variations.

On this graph, stationarity is a projected condition rather than a separate annihilation of the bulk and boundary variations. The first variation closes through the pullback balance
\[
E_L
+
\C^*
\left(
\Pi_L+E_{\partial\B}
\right)
=
0
\quad
\text{in }\VV_{\mathrm{bulk}}',
\]
which gives the variational meaning of bulk--boundary action exchange.

At second order, the open Hessian is the self-adjoint operator associated with the closed quadratic form obtained by restricting the second variation to \(\VV_{\mathrm{op}}\). For pressure-like boundary couplings this gives a pulled-back form, and the Rayleigh--Ritz criterion on the regulated graph determines the threshold at which positivity and coercivity are lost. Boundary degrees of freedom enter not as auxiliary constraints or external forces, but as regulated contributors to stationarity and stability. The framework is not a modification of variational calculus, but a reformulation of admissibility and stationarity on regulated exchange spaces.\\

\backmatter

\bmhead{Supplementary information}

No supplementary information is included in this draft.

\bmhead{Acknowledgements}

The author thanks Jos\'e A. Santiago and Jes\'us Fern\'andez-Castillo for fruitful discussions.

\begin{appendices}

\section{Classical free-boundary problems as separable limits}\label{app:free_boundary}

Classical free-boundary and transversality problems are recovered as separable limits of graph-space stationarity. Consider
\begin{equation}
S[q]
=
\int_{t_0}^{t_1}L(q,\dot q,t)\,\dd t
+
B(q(t_1),t_1),
\label{eq:appendix_free_endpoint_action}
\end{equation}
with \(q(t_0)\) fixed. Its first variation is
\begin{equation}
\delta S
=
\int_{t_0}^{t_1}
\left(
\frac{\partial L}{\partial q}
-
\frac{\dd}{\dd t}
\frac{\partial L}{\partial \dot q}
\right)
\delta q\,\dd t
+
\left[
p(t_1)+\nabla B(q(t_1),t_1)
\right]
\delta q(t_1),
\label{eq:appendix_free_endpoint_variation}
\end{equation}
where \(p=\partial L/\partial\dot q\). If \(\delta q(t_1)\) is independently testable, stationarity gives
\begin{equation}
\frac{\dd}{\dd t}
\frac{\partial L}{\partial \dot q}
-
\frac{\partial L}{\partial q}
=
0,
\qquad
p(t_1)+\nabla B(q(t_1),t_1)=0.
\label{eq:appendix_separable_endpoint_conditions}
\end{equation}
This is the finite-dimensional analogue of
\eqref{eq:bulk_EL}--\eqref{eq:open_boundary_condition}.

If the endpoint variation is constrained by
\begin{equation}
\delta q(t_1)=\C\eta,
\label{eq:appendix_endpoint_compatibility}
\end{equation}
then stationarity gives
\begin{equation}
\int_{t_0}^{t_1}
\left(
\frac{\partial L}{\partial q}
-
\frac{\dd}{\dd t}
\frac{\partial L}{\partial \dot q}
\right)
\eta\,\dd t
+
\left[
p(t_1)+\nabla B(q(t_1),t_1)
\right]
\C\eta
=
0
\quad
\forall \eta.
\label{eq:appendix_projected_endpoint_balance}
\end{equation}
Equivalently,
\begin{equation}
E_L
+
\C^*
\left[
p(t_1)+\nabla B(q(t_1),t_1)
\right]
=
0
\label{eq:appendix_projected_endpoint_adjoint}
\end{equation}
on the admissible variation space. Thus the classical transversality condition is the separable endpoint limit of the regulated balance.

\section{Spherical second variation}\label{app:sphere_second_variation}

For a normal deformation \(r=R+u\) of \(S_R^2\), the area and volume expansions are
\begin{align}
A[u]
&=
4\pi R^2
+
\int_{S_R^2}
\left(
\frac{u^2}{R^2}
+
\frac12|\nabla_\Sigma u|^2
\right)
\dd A
+
O(u^3),
\label{eq:area_expansion_appendix}
\\
V[u]
&=
\frac{4\pi R^3}{3}
+
\int_{S_R^2}u\,\dd A
+
\frac{1}{R}
\int_{S_R^2}u^2\,\dd A
+
O(u^3).
\label{eq:volume_expansion_appendix}
\end{align}
After removal of the mean mode, or under a fixed-volume constraint, the linear volume term vanishes. With the convention \(\delta^2S=(1/2)\Q[u]\), the pressure contribution gives
\begin{equation}
\Q_P^{\partial}[u]
=
\frac{2}{R}
\int_{S_R^2}u^2\,\dd A.
\label{eq:appendix_pressure_form_separable}
\end{equation}
For a regulated boundary displacement \(w=\C u\),
\begin{equation}
\Q_P[u]
=
\Q_P^{\partial}[\C u]
=
\frac{2}{R}
\int_{S_R^2}(\C u)^2\,\dd A.
\label{eq:appendix_pressure_form_pulled_back}
\end{equation}
If
\[
\C Y_{\ell m}=c_\ell Y_{\ell m},
\]
then
\begin{equation}
\Q_P[u]
=
\frac{2}{R}
\sum_{\ell,m}
|c_\ell|^2|u_{\ell m}|^2,
\label{eq:appendix_pressure_form_harmonic}
\end{equation}
which yields the regulated eigenvalues and thresholds in
\eqref{eq:spherical_regulated_eigenvalues}--\eqref{eq:spherical_regulated_threshold}.

\section{Quadratic-form interpretation of the open Hessian}\label{app:quadratic_forms}

Let \(\Q\) be a densely defined, closed, semibounded symmetric quadratic form on a Hilbert space \(\HH\). In the regulated setting, its form domain is the graph space
\begin{equation}
\VV_{\mathrm{op}}
=
\operatorname{Graph}(\C)
=
\{(u,\C u):u\in\VV_{\mathrm{bulk}}\},
\label{eq:appendix_graph_form_domain}
\end{equation}
or an equivalent closed admissible subspace after quotienting null modes. The graph norm and the closedness of the graph are understood in the standard functional-analytic sense \cite{Conway1990,Brezis2011}.

By the representation theorem for closed semibounded quadratic forms, \(\Q\) determines a unique self-adjoint operator \(L\) satisfying \cite{Kato1995,ReedSimon1978}
\begin{equation}
\Q[u,v]
=
\dual{u}{Lv}_{\HH},
\qquad
u\in\VV_{\mathrm{op}},
\quad
v\in\Dom(L)\subset\VV_{\mathrm{op}}.
\label{eq:appendix_representation_theorem}
\end{equation}
The operator domain consists of those \(v\in\VV_{\mathrm{op}}\) for which the map
\[
u\mapsto\Q[u,v]
\]
is continuous in the \(\HH\)-norm.

If the second variation decomposes into bulk, boundary and mixed parts, restriction to the graph gives
\begin{equation}
\Q_{\mathrm{op}}[u]
=
\Q_{\mathrm{bulk}}[u]
+
\Q_{\partial}[\C u]
+
2\Q_{\mathrm{mix}}[u,\C u].
\label{eq:appendix_general_graph_form}
\end{equation}
In the additive pressure-like case,
\begin{equation}
\Q_\Pi[u]
=
\Q_G[u]
-
\Pi\Q_P^{\partial}[\C u].
\label{eq:appendix_pressure_graph_form}
\end{equation}
Thus
\[
\Leff=\LG-\Pi\LP
\]
is an operator representation of the pulled-back graph form, not an unconstrained subtraction of independent bulk and boundary operators.

\end{appendices}

\section*{Declarations}

\bmhead{Funding}
Financial support was provided by the Agencia Estatal de Investigación (AEI, Spain)
under Grants TED2021-132296B-C52 and CPP2024-011880, and by Fundación
BBVA--Programa Fundamentos 2025.

\bmhead{Competing interests}
The author declares no competing interests.

\bmhead{Ethics approval and consent to participate}
Not applicable.

\bmhead{Consent for publication}
Not applicable.

\bmhead{Data availability}
No datasets were generated or analysed during the current study.

\bmhead{Materials availability}
Not applicable.

\bmhead{Code availability}
Not applicable.

\bmhead{Author contribution}
F.M. conceived the framework, developed the mathematical formulation and wrote the manuscript.

\bibliography{bibliography-CMP}

\end{document}